\RequirePackage{fix-cm}
\documentclass[twocolumn]{svjour3}       
\smartqed  
\usepackage{graphicx}
\usepackage [latin1]{inputenc}
\usepackage{marvosym}
\usepackage{amssymb}
\usepackage{indentfirst}
\usepackage{graphicx}
\usepackage{epstopdf}
\usepackage{float}
\usepackage{subfig}
\usepackage{amssymb}
\usepackage{stmaryrd}
\usepackage{amsfonts}
\usepackage{amsmath}
\usepackage{latexsym}
\usepackage{color}
\usepackage{CJK}
\usepackage{mathrsfs}
\usepackage{algorithm} 
\usepackage{algorithmic} 
\usepackage{multirow} 
\usepackage{amsmath}
\usepackage{xcolor}
\usepackage{graphicx}
\usepackage{mathptmx}
\usepackage{indentfirst}
\usepackage{booktabs}
\usepackage{float}
\usepackage{subfig}
\usepackage{amssymb}
\usepackage{stmaryrd}
\usepackage{dsfont}
\usepackage{amssymb}
\usepackage{amsfonts}
\usepackage{amsmath}
\usepackage{latexsym}
\usepackage{color}
\usepackage{cite}
\usepackage{mathrsfs}
\usepackage{epstopdf}
\journalname{Nonlinear Dyn}
\begin{document}

\title{Forecast analysis of the epidemics trend of COVID-19 in the United States by a generalized fractional-order SEIR model\thanks{This work is supported by the National Nature Science Foundation of China under Grant 61772063, Beijing Natural Science Foundation under Grant Z180005.
}}


\author{Conghui Xu        \and
         Yongguang Yu* \and
         YangQuan Chen* \and
         Zhenzhen Lu
}


\institute{
Conghui Xu. Yongguang Yu (\Letter). Zhenzhen Lu. \\
Department of Mathematics, School of Science, Beijing Jiaotong University, Beijing 100044, China.\\
\email{ygyu@bjtu.edu.cn}\\
YangQuan Chen (\Letter)\\
Mechatronics, Embedded Systems and Automation Lab, University of California, Merced, CA 95343, USA\\
\email{ychen53@ucmerced.edu}}
\date{Received: date / Accepted: date}

\maketitle

\begin{abstract}
In this paper, a generalized fractional-order SEIR model is proposed, denoted by SEIQRP model, which has a basic guiding significance for the prediction of the possible outbreak of infectious diseases like COVID-19 and other insect diseases in the future. Firstly, some qualitative properties of the model are analyzed. The basic reproduction number $\mathnormal{R}_{0}$ is derived. When $\mathnormal{R}_{0}<1$, the disease-free equilibrium point is unique and locally asymptotically stable. When $\mathnormal{R}_{0}>1$, the endemic equilibrium point is also unique. Furthermore, some conditions are established to ensure the local asymptotic stability of disease-free and endemic equilibrium points. The trend of COVID-19 spread in the United States is predicted. Considering the influence of the individual behavior and government mitigation measurement, a modified SEIQRP model is proposed, defined as SEIQRPD model. According to the real data of the United States, it is found that our improved model has a better prediction ability for the epidemic trend in the next two weeks. Hence, the epidemic trend of the United States in the next two weeks is investigated, and the peak of isolated cases are predicted. The modified SEIQRP model successfully capture the development process of COVID-19, which provides an important reference for understanding the trend of the outbreak.
\end{abstract}

\section{Introduction}
\label{intro}
The outbreak of the coronavirus disease (COVID-19) in 2019 occurred in Wuhan, China, at the end of 2019. This is a severe respiratory syndrome caused by the novel coronavirus of zoonotic origin\cite{1}. The Chinese government has implemented many measures, including the establishment of specialized hospitals and restrictions on travel, to reduce the spread. By April 20, 2020, the outbreak in China has been basically controlled. However, the outbreak is still rampant all over the world. At present, the United States, Italy, Spain and other countries are still in the rising stage of the outbreak. It has posed a great threat to the public health and safety of the world.

At present, many countries have adopted mitigation measures to restrict travel and public gatherings, which have a serious impact on the economy. Therefore, it is very important to predict the development trend of this epidemic and estimate the peak of the isolated cases. The epidemic model is a basic tool to research the dynamic behaviours of disease and predict the spreading trend of disease. Establishing a reasonable epidemic model can effectively characterize the development process of the disease. Therefore, a generalized SEIR epidemic model is proposed in this paper, and is defined as SEIQRP model. Some qualitative properties of this model are first analyzed, including the existence and uniqueness of the disease-free and the endemic equilibrium points. Then, conditions are also established to ensure the local asymptotic stability of both disease-free and endemic equilibrium points.

So far, many scholars have researched COVID-19 from different perspectives \cite{1,2,3,4,5,23}. In \cite{6}, the epidemics trend of COVID-19 in China was predicted under public health interventions. In \cite{7}, the basic reproduction number of COVID-19 in China was estimated and the data-driven analysis was performed in the early phase of the outbreak. In order to predict the cumulative number of confirmed cases and combine the actual measures taken by the government on the outbreak of COVID-19, we further put forward an improved SEIQRP model, denoted by SEIQRPD model. At present, we can obtain the epidemic data of COVID-19 outbreak in the United States before April 20, 2020. We use these data and the improved model to predict the epidemic trend of the United States in the next two weeks, and estimate the peak of isolated cases. Firstly, the data before April 5 was selected to identify the model parameters, and the prediction ability of the improved model with the epidemiological data from April 6 to April 20 is verified. Then, the cumulative number of confirmed cases and isolated cases are predicted in the next two weeks. The peak of isolated cases is thus predicted.

In recent years, with the continuous development of fractional calculus theory, fractional order systems modeling approaches have been applied in various engineering and non-engineering fields \cite{19,20,21}. Compared with the short memory of the integer derivative, the fractional derivative has the information of the whole time interval or long memory. It is more accurate to describe the biological behavior of population by using fractional differential equation.

The rest of this manuscript is structured as below. In Section 2, the fractional SEIQRP model and the modified SEIQRP model is proposed. Some qualitative properties of the SEIQRP model are discussed in Section 3. In Section 4, the prediction ability of the modified SEIQRP model is verified by using real data. The disease development in the United States in the next two weeks after April 21, 2020 is predicted. Finally, conclusion and some future works are discussed in Section 5.

\section{Preliminaries and Model Derivation}
\label{sec:1}
\subsection{Preliminaries}
\label{sec:2}
In this section, some useful lemmas and definitions will be given to analyze some results of this paper.

\begin{definition} \cite{8} The Caputo fractional order derivative is given below
\[_{0}^{C}D^{\alpha}_{t}f(t)=\frac{1}{\Gamma(n-\alpha)}\int_{0}^{t}(t-\xi)^{n-\alpha-1}f^{(n)}(\xi)d\xi,\]
where $n-1\leq\alpha\leq n$. In order to simplify the symbolic representation, the Caputo fractional differential operator $_{0}^{C}D^{\alpha}_{t}$ in this paper are represented by $D^{\alpha}$.
\end{definition}
\begin{definition} \cite{8,22} The Mittag-Leffler function is given below
\[E_{\alpha}(z)=\sum_{k=0}^{\infty}\frac{z^{k}}{\Gamma(k\alpha+1)},\]
where $n-1\leq\alpha\leq n$.
\end{definition}
Consider the following $n$-dimensional fractional order differential equation system
\begin{equation}
D^{\alpha}X(t)=AX(t),\;\;x(0)=x_{0},
\end{equation}
where $\alpha\in(0,1)$, $X(t)=(x_{1}(t),x_{2}(t),\dots,x_{n}(t))^{T}$ is an $n$ dimensional state vector and $A$ is an $n\times n$ constant matrix.
\begin{lemma}\cite{9} For the corresponding linear time-invariant system (1), the following results are true:\\
(i) The zero solution is asymptotically stable, if and only if all eigenvalues $s_{j}$ $(j=1,2,\dots,n)$ of $A$ satisfy $|\arg(s_{j})|>\frac{\alpha\pi}{2}$.\\
(ii) The zero solution is stable, if and only if all eigenvalues $s_{j}$ of $A$ satisfy $|\arg(s_{j})|\geq\frac{\alpha\pi}{2}$ and eigenvalues with $|\arg(s_{j})|=\frac{\alpha\pi}{2}$ have the same algebraic multiplicity and geometric multiplicity.
\end{lemma}
\subsection{SEIQRP model}
\label{sec:3}
\begin{figure}
    \begin{center}
        \includegraphics[width=1\linewidth]{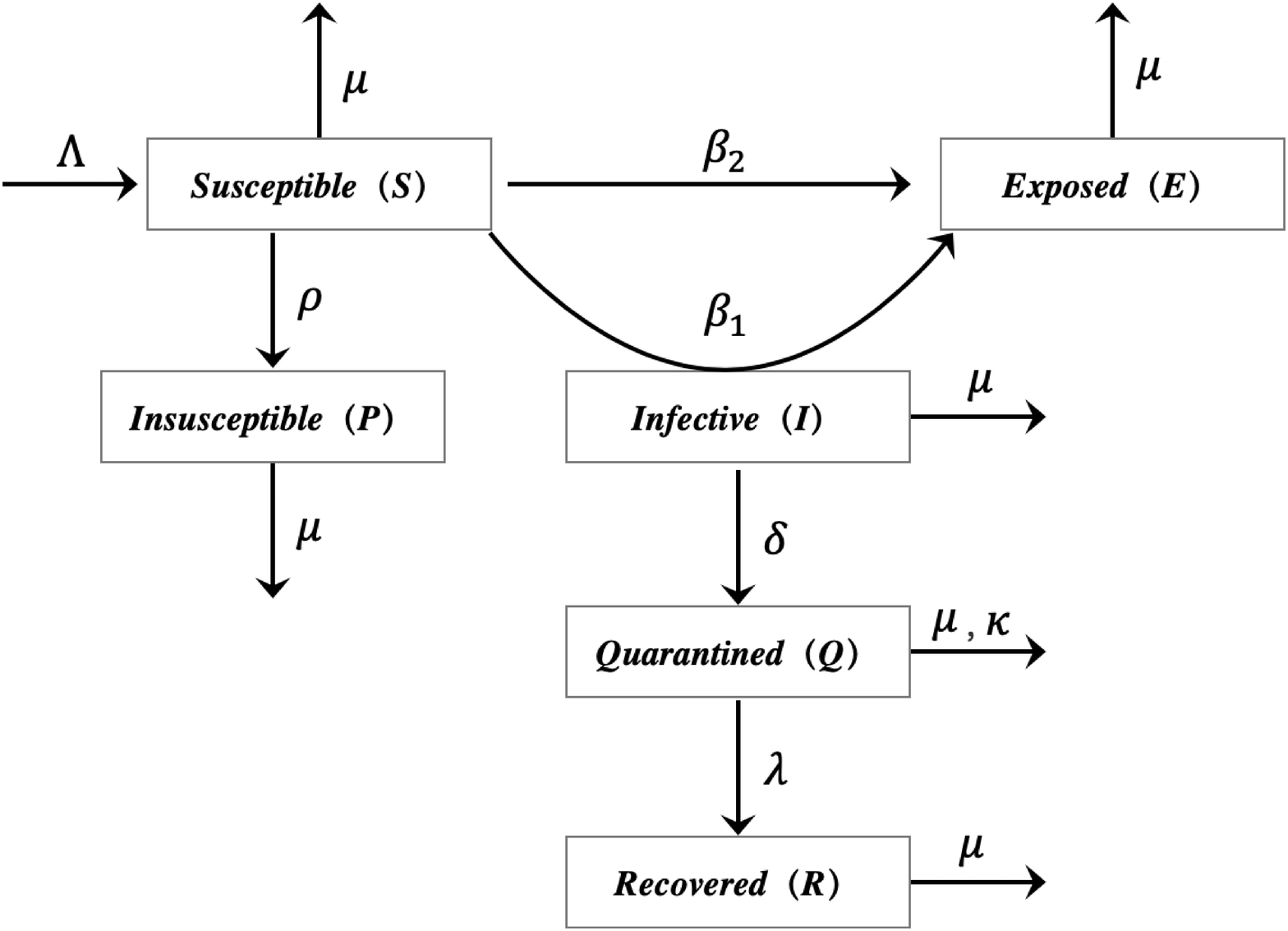}
       \caption{ Transmission diagram for the model involving six population classes.}
    \end{center}
    \end{figure}
The outbreak of COVID-19 has had a great impact on the economic growth of any country and daily life of any human. In order to control and prevent the possible outbreak of infectious diseases like the COVID-19 or other insect diseases in the future, it is very important to establish an appropriate model. The transmission diagram of the generalized SEIR model proposed in this paper is shown in Figure 1. We divide the total population into six distinct epidemic classes: susceptible, exposed, infectious, quarantined, recovered and insusceptible. We will represent the number of individuals at time $t$ in the above classes by $S(t)$, $E(t)$, $I(t)$, $Q(t)$, $R(t)$ and $P(t)$, respectively. The specific explanations of the above six categories are as follows:
\begin{itemize}
\item[$\bullet$] Susceptible $S(t)$: the number of uninfected individuals at the time $t$.
\item[$\bullet$] Exposed $E(t)$: the number of infected individuals at the time $t$ but still in incubation period (without clinical symptoms and low infectivity).
\item[$\bullet$] Infectious $I(t)$: the number of infected individuals at the time $t$ (with obvious clinical symptoms).
\item[$\bullet$] Quarantined $Q(t)$: the number of individuals who have been diagnosed and isolated at the time $t$.
\item[$\bullet$] Recovered $R(t)$: the number of recovered individuals at the time $t$.
\item[$\bullet$] Insusceptible $P(t)$: the number of susceptible individuals who are not exposed to the external environment at the time $t$.
\end{itemize}

 The incidence rate plays a very important role in the epidemic, which can describe the evolution of infectious disease. According to the spread of different diseases in different regions, there are many forms of incidence rate.\cite{10,11,12}. We expect to establish a generalized incidence rate, which can contain most of these specific forms. Therefore, the SEIQRP model with general incidence rate is given by
 \begin{eqnarray}
 \left\{\begin{split}
D^{\alpha}S(t)=&\Lambda-\beta_{1}f_{1}(S(t))g_{1}(I(t))-\beta_{2}f_{2}(S(t))g_{2}(E(t))\\&-\mu S(t)-\rho S(t),\\
D^{\alpha}E(t)=&\beta_{1}f_{1}(S(t))g_{1}(I(t))+\beta_{2}f_{2}(S(t))g_{2}(E(t))\\&-\epsilon E(t)-\mu E(t),\\
D^{\alpha}I(t)=&\epsilon E(t)-(\delta+\mu)I(t),\\
D^{\alpha}Q(t)=&\delta I(t)-(\lambda+\kappa+\mu)Q(t),\\
D^{\alpha}R(t)=&\lambda Q(t)-\mu R(t),\\
D^{\alpha}P(t)=&\rho S(t)-\mu P(t),
  \end{split}
 \right.
\end{eqnarray}
where $\Lambda$ is the inflow number of susceptible individuals, $\beta_{1}$ and $\beta_{2}$ denote the infection rates of the infected individuals and the exposed individuals, respectively, $\rho$ is the protection rate, $\epsilon$ represent the incubation rate. $\delta$ is the rate at which symptomatic infections are diagnosed and quarantined. $\lambda$ and $\kappa$ represent the cure rate of isolated individuals and the death rate caused by the disease, respectively, $\mu$ is the natural mortality, the incidence rate $\beta_{1}f_{1}(S)g_{1}(I)$, $\beta_{2}f_{2}(S)g_{2}(E)$ are used to describe the transmission of diseases, which satisfy the following conditions\cite{16}:\\
(i) $f_{1}(0)=f_{2}(0)=g_{1}(0)=g_{2}(0)=0$,\\
(ii) $f_{1}(S)>0$, $f_{2}(S)>0$, $g_{1}(I)>0$, $g_{2}(E)>0$, for any $S$, $E$, $I>0$,\\
(iii) $f'_{1}(S)>0$, $f'_{2}(S)>0$, $g'_{1}(I)>0$, $g'_{2}(E)>0$, for any $S$, $E$, $I>0$,\\
(iv) $g'_{1}(I)-\frac{g_{1}(I)}{I}\leq 0$, $g'_{2}(E)-\frac{g_{2}(E)}{E}\leq 0$, for any $E$, $I>0$.\\

\subsection{Modified SEIQRP model}
\label{sec:4}
In the following, we will combine the actual situation and government mitigation policy to improve the SEIQRP model (2). Firstly, during the outbreak of COVID-19, we need to make reasonable assumptions according to the actual mitigation policies and circumstances.
\begin{itemize}
\item[$\bullet$] During the COVID-19 outbreak, population mobility is strictly controlled by many countries. Especially, the policy of city closure was implemented in Hubei province, China. Therefore, the impact of migration will not be considered in the improved model.
\item[$\bullet$] For the prediction of this short-term virus transmission, the impact of natural mortality will not be considered. It is assumed that novel coronavirus is the only cause of death during the outbreak.
\item[$\bullet$] In order to predict the trend of cumulative confirmed cases, it is necessary to simulate the number of death cases. Therefore, a new class $D(t)$ to the SEIQRP model will be added in the modified model, which denotes the number of death cases at time $t$.
\item[$\bullet$] The time-varying cure rate, mortality rate and infection rate will be applied to the improved model. This can better simulate the impact of the improvement of medical conditions and government control on individuals in reality.
\end{itemize}
Based on the above assumptions and analysis, the following improved model for COVID-19 is proposed.
 \begin{eqnarray}
 \left\{\begin{split}
D^{\alpha}S(t)=&-\beta_{1}(t)f_{1}(S(t))g_{1}(I(t))-\beta_{2}f_{2}(S(t))g_{2}(E(t))\\&-\rho S(t),\\
D^{\alpha}E(t)=&\beta_{1}(t)f_{1}(S(t))g_{1}(I(t))+\beta_{2}f_{2}(S(t))g_{2}(E(t))\\&-\epsilon E(t),\\
D^{\alpha}I(t)=&\epsilon E(t)-\delta I(t),\\
D^{\alpha}Q(t)=&\delta I(t)-(\lambda(t)+\kappa(t))Q(t),\\
D^{\alpha}R(t)=&\lambda(t) Q(t),\\
D^{\alpha}P(t)=&\rho S(t),\\
D^{\alpha}D(t)=&\kappa(t) Q(t),
  \end{split}
 \right.
\end{eqnarray}
where $\beta_{1}(t)=\sigma_{1}\exp{(-\sigma_{2} t)}$, $\lambda(t)=\lambda_{1}(1-\exp{(\lambda_{2}t)})$ and $\kappa(t)=\kappa_{1}\exp{(-\kappa_{2}t)}$. The parameters $\sigma_{1}$, $\sigma_{2}$, $\lambda_{1}$, $\lambda_{2}$, $\kappa_{1}$ and $\kappa_{2}$ are all positive, where $\sigma$ is used to simulate the intensity of government control. It should be emphasized that the protection rate $\rho$ for susceptible individuals also reflects the intensity of government control.
\section{Qualitative Analysis of the SEIQRP Model}
\label{sec:5}
\subsection{The Existence and Uniqueness of Equilibrium Point}
\label{sec:6}
Obviously, the right side of system (2) satisfies the local Lipschitz condition, then there exists a unique solution of system (2) \cite{8,15}.
\begin{theorem}
The solutions of system (2) are non-negative, and the closed set $\Omega=\{(S,E,I,Q,R,P)\in \mathbb{R}_{+}^{6}: 0\leq S+E+I+Q+R+D\leq\frac{\Lambda}{\mu}\}$ is a positive invariant set of system (2).
\end{theorem}
\begin{proof}
In order to investigate the non-negativity of solutions of the system (2), we consider the following system
 \begin{eqnarray}
 \left\{\begin{split}
D^{\alpha}S_{1}(t)=&-\beta_{1}f_{1}(S(t))g_{1}(I(t))-\beta_{2}f_{2}(S(t))g_{2}(E(t))\\&-\mu S(t)-\rho S(t),\\
D^{\alpha}E_{1}(t)=&\beta_{1}f_{1}(S(t))g_{1}(I(t))+\beta_{2}f_{2}(S(t))g_{2}(E(t))\\&-\epsilon E(t)-\mu E(t),\\
D^{\alpha}I_{1}(t)=&\epsilon E(t)-(\delta+\mu)I(t),\\
D^{\alpha}Q_{1}(t)=&\delta I(t)-(\lambda+\kappa+\mu)Q(t),\\
D^{\alpha}R_{1}(t)=&\lambda Q(t)-\mu R(t),\\
D^{\alpha}P_{1}(t)=&\rho S(t)-\mu P(t),
  \end{split}
 \right.
\end{eqnarray}
where the initial conditions are $S_{10}=0$, $E_{10}=0$, $I_{10}=0$, $Q_{10}=0$, $R_{10}=0$ and $P_{10}=0$. $(0,0,0,0,0,0)$ is the unique solution of the above system. According to the fractional-order comparison theorem \cite{17}, one can deduce that the solutions of system (2) satisfy $S(t)\geq0$, $E(t)\geq0$, $I(t)\geq0$, $Q(t)\geq0$, $R(t)\geq0$ and $P(t)\geq0$.
By adding six equations of system (2), one can deduce
\[
\begin{split}
D^{\alpha}N(t)&=\Lambda-\mu N(t)-\kappa Q(t)\\&\leq\Lambda-\mu N(t).
\end{split}
\]
By applying the fractional-order comparison theorem, one has
\[D^{\alpha}N(t)\leq(-\frac{\Lambda}{\mu}+N(0))E_{\alpha}(-\mu t^{\alpha})+\frac{\Lambda}{\mu}.\]
When $N(0)\leq\frac{\Lambda}{\mu}$, since $E_{\alpha}(-\mu t^{\alpha})\geq 0$, we have
\[S(t)+E(t)+I(t)+Q(t)+R(t)+P(t)\leq\frac{\Lambda}{\mu}\]
Thus we can draw the result of Theorem 1.
\end{proof}
System (2) has an obvious disease-free equilibrium point $M_{0}=(S_{0},0,0,0,0,P_{0})$, where
\[S_{0}=\frac{\Lambda}{\mu+\rho},\;\;P_{0}=\frac{\rho\Lambda}{\mu(\mu+\rho)}.\]
In order to obtain the endemic equilibrium point of the system (2), we set:
 \begin{equation}
 \left\{\begin{split}
&\Lambda-\beta_{1}f_{1}(S(t))g_{1}(I(t))-\beta_{2}f_{2}(S(t))g_{2}(E(t))\\&-\mu S(t)-\rho S(t)=0,\\
&\beta_{1}f_{1}(S(t))g_{1}(I(t))+\beta_{2}f_{2}(S(t))g_{2}(E(t))\\&-\epsilon E(t)-\mu E(t)=0,\\
&\epsilon E(t)-(\delta+\mu)I(t)=0,\\
&\delta I(t)-(\lambda+\kappa+\mu)Q(t)=0,\\
&\lambda Q(t)-\mu R(t)=0,\\
&\rho S(t)-\mu P(t)=0,
  \end{split}
 \right.
\end{equation}
which implies
\[S=\frac{\Lambda-(\epsilon+\mu)E}{\mu+\rho},\;\;I=\frac{\epsilon E}{\sigma+\mu}.\]
Combining the above equations and the second equation of (5), one has
 \begin{equation}
\begin{split}
 (\mu+\epsilon)E=&\beta_{1}f_{1}(\frac{\Lambda-(\mu+\epsilon)E}{\mu+\rho})g_{1}(\frac{\epsilon E}{\mu+\delta})\\&+\beta_{2}f_{2}(\frac{\Lambda-(\mu+\epsilon)E}{\mu+\rho})g_{2}(E).
   \end{split}
\end{equation}
Define
 \begin{equation}
\begin{split}
 \varphi(E)=&\beta_{1}f_{1}(\frac{\Lambda-(\mu+\epsilon)E}{\mu+\rho})g_{1}(\frac{\epsilon E}{\mu+\delta})\\&+\beta_{2}f_{2}(\frac{\Lambda-(\mu+\epsilon)E}{\mu+\rho})g_{2}(E)-(\mu+\epsilon)E.
   \end{split}
\end{equation}
Note that $\varphi(0)=0$ and $\varphi(\frac{\Lambda}{\mu+\epsilon})=-\Lambda<0$. In order to show that $\varphi(E)=0$ has at least one positive root in the interval $(0,\frac{\Lambda}{\mu+\epsilon}]$, we need to prove that $\varphi'(0)>0$. Hence
 \begin{equation}
\begin{split}
 \varphi'(E)=&-\frac{(\mu+\epsilon)\beta_{1}}{\mu+\rho}f'_{1}(\frac{\Lambda-(\mu+\epsilon)E}{\mu+\rho})g_{1}(\frac{\epsilon E}{\mu+\delta})\\&
 +\frac{\epsilon \beta1}{\mu+\delta}f_{1}(\frac{\Lambda-(\mu+\epsilon)E}{\mu+\rho})g'_{1}(\frac{\epsilon E}{\mu+\delta})\\&
 -\frac{(\mu+\epsilon)\beta_{2}}{\mu+\rho}f'_{2}(\frac{\Lambda-(\mu+\epsilon)E}{\mu+\rho})g_{2}(E)\\&
 +\beta_{2}f_{2}(\frac{\Lambda-(\mu+\epsilon)E}{\mu+\rho})g'_{2}(E)-(\mu+\epsilon).
   \end{split}
\end{equation}
Therefore, we have
 \[
\begin{split}
 \varphi'(0)&=\frac{\epsilon\beta_{1}}{\mu+\delta}f_{1}(S_{0})g'_{1}(0)+\beta_{2}f_{2}(S_{0})g'_{2}(0)-(\mu+\epsilon)\\&=(\mu+\epsilon)(\mathnormal{R}_{0}-1),
   \end{split}
\]
where the basic reproduction number is given by
 \begin{equation}
\mathnormal{R}_{0}=\frac{\epsilon\beta_{1}}{(\mu+\epsilon)(\mu+\delta)}f_{1}(S_{0})g'_{1}(0)+\frac{\beta_{2}}{\mu+\epsilon}f_{2}(S_{0})g'_{2}(0).
\end{equation}
If $\mathnormal{R}_{0}>1$, the system (2) has at least one endemic equilibrium point $M_{*}=(S_{*},E_{*},I_{*},Q_{*},R_{*},P_{*})$, where
\[S_{*}=\frac{\Lambda-(\mu+\epsilon)E_{*}}{\mu+\rho},\;\;I_{*}=\frac{\epsilon E_{*}}{\mu+\delta},\;\;Q_{*}=\frac{\delta I_{*}}{\mu+\lambda+\kappa},\]
\[R_{*}=\frac{\lambda Q_{*}}{\mu},\;\;P_{*}=\frac{\rho S_{*}}{\mu}.\]
In the following section, we show that endemic equilibrium point $M_{*}$ is unique. According to the above analysis and hypothesis (i)-(iii), one has
 \begin{equation}
\begin{split}
 \varphi'(E_{*})=&-\frac{(\mu+\epsilon)\beta_{1}}{\mu+\rho}f'_{1}(S_{*})g_{1}(I_{*})+\frac{\epsilon\beta_{1}}{\mu+\delta}f_{1}(S_{*})g'_{1}(I_{*})\\&
 -\frac{(\mu+\epsilon)\beta_{2}}{\mu+\rho}f'_{2}(S_{*})g_{2}(E_{*})+\beta_{2}f_{2}(S_{*})g'_{2}(E_{*})\\&-(\mu+\epsilon)\\
 =&-\frac{(\mu+\epsilon)\beta_{1}}{\mu+\rho}f'_{1}(S_{*})g_{1}(I_{*})-\frac{(\mu+\epsilon)\beta_{2}}{\mu+\rho}f'_{2}(S_{*})g_{2}(E_{*})\\&
 +\frac{\epsilon\beta_{1}}{\mu+\delta}f_{1}(S_{*})g'_{1}(I_{*})-\frac{\epsilon\beta_{1}f_{1}(S_{*})g_{1}(I_{*})}{(\mu+\delta)I_{*}}\\&
 +\beta_{2}f_{2}(S_{*})g'_{2}(E_{*})-\frac{\beta_{2}f_{2}(S_{*})g_{2}(E_{*})}{E_{*}}\\
 =&-\frac{(\mu+\epsilon)\beta_{1}}{\mu+\rho}f'_{1}(S_{*})g_{1}(I_{*})-\frac{(\mu+\epsilon)\beta_{2}}{\mu+\rho}f'_{2}(S_{*})g_{2}(E_{*})\\&
 +\frac{\epsilon\beta_{1}}{\mu+\delta}f_{1}(S_{*})[g'_{1}(I_{*})-\frac{g_{1}(I_{*})}{I_{*}}]\\&
 +\beta_{2}f_{2}(S_{*})[g'_{2}(E_{*})-\frac{g_{2}(E_{*})}{E_{*}}].
   \end{split}
\end{equation}
By the hypothesis (iv), this implies $\varphi'(E_{*})<0$. If there is another endemic equilibrium point $M_{**}$, then $\varphi'(E_{**})\geq 0$ holds, which contradicts the previous discussion. Hence, system (2) has a unique endemic equilibrium point $M_{*}$ when $\mathnormal{R}_{0}>1$. Based on the above analysis, the following result can be obtained.
\begin{theorem}
System (2) has a unique disease-free equilibrium point $M_{0}$, if $\mathnormal{R}_{0}<1$. System (2) has a unique endemic equilibrium point $M_{*}$, if $\mathnormal{R}_{0}>1$.
\end{theorem}
\subsection{Stability Analysis}
\label{sec:7}
In this section, the local asymptotic stability of disease-free equilibrium point $M_{0}$ and endemic equilibrium point $M_{*}$ for system (2) is discussed.
\begin{theorem}
With regard to system (2), the disease-free equilibrium point $M_{0}$ is locally asymptotic stability, if $\mathnormal{R_{0}}<1$; the disease-free equilibrium point $M_{0}$ is unstable, if $\mathnormal{R_{0}}>1$.
\end{theorem}
\begin{proof}
The Jacobian matrix of the system (2) at the disease-free equilibrium point $M_{0}$ is
\[
J_{M_{0}}=
 \begin{bmatrix}
   J_{11} & J_{12}\\
   J_{21} & J_{22}
  \end{bmatrix},
\]
where

\[J_{11}=
 \begin{bmatrix}
   -\mu-\rho & -\beta_{2}f_{2}(S_{0})g'_{2}(0) & -\beta_{1}f_{1}(S_{0})g'_{1}(0) \\
   0 & \beta_{2}f_{2}(S_{0})g'_{2}(0)-\mu-\epsilon & \beta_{1}f_{1}(S_{0})g'_{1}(0) \\
   0 & \epsilon & -\mu-\delta
  \end{bmatrix},\]
\[
J_{12}=
\begin{bmatrix}
   0 & 0 & 0 \\
   0 & 0 & 0 \\
   0 & 0 & 0
  \end{bmatrix},\;\;
J_{21}=
 \begin{bmatrix}
   0 & 0 & \delta \\
   0 & 0 & 0 \\
   \rho & 0 & 0
  \end{bmatrix},
\]
\[
J_{22}=
 \begin{bmatrix}
 -\mu-\kappa-\lambda & 0 & 0\\
   \lambda & -\mu & 0\\
   0 & 0 & -\mu
  \end{bmatrix}.
\]
The corresponding characteristic equation is
 \begin{equation}
\begin{split}
H(s)=&|sE-J_{M_{0}}|\\=&(s+\mu)^{2}(s+\mu+\rho)(s+\mu+\lambda+\kappa)H_{1}(s),
   \end{split}
\end{equation}
where
 \begin{equation}
\begin{split}
H_{1}(s)=&s^{2}+(\delta+2\mu+\epsilon-\beta_{2}f_{2}(S_{0})g'_{2}(0))s\\&+(\mu+\delta)(\mu+\epsilon)-\beta_{2}(\mu+\delta)f_{2}(S_{0})g'_{2}(0)\\&-\epsilon\beta_{1}f_{1}(S_{0})g'_{1}(0).
   \end{split}
\end{equation}
The characteristic equation $H(s)=0$ has four obvious negative characteristic roots, which are denoted by $s_{1}=s_{2}=-\mu$, $s_{3}=-\mu-\rho$ and $s_{4}=-\mu-\lambda-\kappa$, respectively. The discriminant of $H_{1}(s)$ in quadratic form is
\[
\begin{split}
\Delta=&[\delta+2\mu+\epsilon-\beta_{2}f_{2}(S_{0})g'_{2}(0)]^2-4[(\mu+\delta)(\mu+\epsilon)\\&-\beta_{2}(\mu+\delta)f_{2}(S_{0})g'_{2}(0)-\epsilon\beta_{1}f_{1}(S_{0})g'_{1}(0)]\\
=&(\delta-\epsilon+\beta_{2}f_{2}(S_{0})g'_{2}(0))^{2}+4\epsilon\beta_{1}f_{1}(S_{0})g'_{1}(0)\\>&0.
   \end{split}
\]
This implies that the other two eigenvalues $s_{5}$ and $s_{6}$ of characteristic equation $H(s)=0$ are real roots. Hence
\[s_{5}+s_{6}=-(\delta+2\mu+\epsilon-\beta_{2}f_{2}(S_{0})g'_{2}(0)),\]
\[s_{5}s_{6}=(\mu+\delta)(\mu+\epsilon)(1-\mathnormal{R}_{0}).\]
If $\mathnormal{R}_{0}<1$, then one can obtain $s_{5}+s_{6}<0$ and $s_{5}s_{6}>0$, which imply $s_{5}<0$ and $s_{6}<0$. If $\mathnormal{R}_{0}>1$, one has $s_{5}s_{6}<0$, which imply $s_{5}>0$ or $s_{6}>0$. It follows from Lemma 1 that the proof is completed.
\end{proof}

Further, we will show the locally asymptotic stability of the endemic equilibrium point $M_{*}$ of system (2). Similarly, the corresponding Jacobian matrix of the system (2) at $M_{*}$ is
\[
J_{M_{*}}=
 \begin{bmatrix}
   J_{1} & J_{2}\\
   J_{3} & J_{4}
  \end{bmatrix},
\]
where
\[
J_{1}=
 \begin{bmatrix}
   -l_{3}-l_{4}-\mu-\rho & -l_{2} & -l_{1} \\
   l_{3}+l_{4} & l_{2}-\epsilon-\mu & l_{1}\\
   0 & \epsilon & -\delta-\mu
  \end{bmatrix},
\]
\[
J_{2}=
 \begin{bmatrix}
    0 & 0 & 0\\
   0 & 0 & 0\\
    0 & 0 & 0\\
  \end{bmatrix},\;\;
J_{3}=
 \begin{bmatrix}
   0 & 0 & \delta \\
   0 & 0 & 0 \\
   \rho & 0 & 0 \\
  \end{bmatrix},
\]
\[
J_{4}=
 \begin{bmatrix}
  -\mu-\lambda-\kappa & 0 & 0\\
  \lambda & -\mu & 0\\
  0 & 0 & -\mu\\
  \end{bmatrix},
\]
with
\[l_{1}=\beta_{1}f_{1}(S_{*})g'_{1}(I_{*}),\;\;l_{2}=\beta_{2}f_{2}(S_{*})g'_{2}(E_{*}),\]
\[l_{3}=\beta_{1}f'_{1}(S_{*})g_{1}(I_{*}),\;\;l_{4}=\beta_{2}f'_{2}(S_{*})g_{2}(E_{*}).\]
Hence, the corresponding characteristic equation is
 \begin{equation}
\begin{split}
L(s)=&|sE-J_{M_{*}}|\\=&(s+\mu)^{2}(s+\mu+\lambda+\kappa)L_{1}(s),
   \end{split}
\end{equation}
where
\[L_{1}(s)=s^{3}+a_{1}s^{2}+a_{2}s+a_{3},\]
with
\[a_{1}=\epsilon+3\mu+\delta+\rho-l_{2}+l_{3}+l_{4},\]
\[\begin{split}
a_{2}=&(\mu+\epsilon-\l_{2})(2\mu+\delta+\rho)+(\mu+\delta)(\mu+\rho)\\&+(2\mu+\delta+\epsilon)(l_{3}+l_{4})-\epsilon l_{1},
\end{split}\]
\[\begin{split}
a_{3}=&(\mu+\delta)(\mu+\rho)(\mu+\epsilon-\l_{2})+(\mu+\delta)(\mu+\epsilon)(l_{3}+l_{4})\\&-\epsilon(\mu+\rho)l_{1}.
\end{split}\]
By the corresponding results in \cite{13,14}, let
\begin{equation}
\begin{split}
\mathnormal{D}_{1}(L_{1}(s))=&
{\left| \begin{array}{ccccc}
1 & a_{1} & a_{2} & a_{3} & 0\\
0 & 1 & a_{1} & a_{2} & a_{3}\\
3 & 2la_{1} & a_{2} & 0 & 0\\
0 & 3 & 2a_{1} & a_{2} & 0\\
0 & 0 & 3 & 2a_{1} & a_{2}
\end{array} \right|}\\=&18a_{1}a_{2}a_{3}+a_{1}^{2}a_{2}^{2}-4a_{1}^{3}a_{3}-4a_{2}^{3}-27a_{3}^{2}.
\end{split}
\end{equation}
Then, the following result can be obtained.
\begin{theorem}
With regard to system (2), assume that $\mathnormal{R}_{0}>1$,
(i) If $\mathnormal{D}_{1}(L_{1}(s))>0$, $a_{1}>0$, $a_{3}>0$ and $a_{1}a_{2}>a_{3}$, then the endemic equilibrium point $M_{*}$ is locally asymptotically stable.\\
(ii) If $\mathnormal{D}_{1}(L_{1}(s))<0$, $a_{1}>0$, $a_{2}>0$ and $a_{3}>0$, then the endemic equilibrium point $M_{*}$ is locally asymptotically stable for $\alpha\in(0,\frac{2}{3})$.\\
(iii) If $\mathnormal{D}_{1}(L_{1}(s))<0$, $a_{1}<0$ and $a_{2}<0$, then the endemic equilibrium point $M_{*}$ is unstable for $\alpha\in(\frac{2}{3},1)$.\\
(iv) If $\mathnormal{D}_{1}(L_{1}(s))<0$, $a_{1}>0$, $a_{2}>0$ and $a_{1}a_{2}=a_{3}$, then for $\alpha\in(0,1)$, the endemic equilibrium point $M_{*}$ is locally asymptotically stable.
\end{theorem}
\begin{proof}
Based on the previous discussion, the characteristic equation $L(s)=0$ has three obvious negative roots $s_{1}=s_{2}=-\mu$ and $s_{3}=-\mu-\lambda-\kappa$. In order to investigate the stability of equilibrium point $M_{*}$, we only need to discuss the range of the root of $L_{1}(s) = 0$, denoted by $s_{4}$, $s_{5}$ and $s_{6}$.

(i) By the results in \cite{14}, if $\mathnormal{D}_{1}(L_{1}(s))>0$, then $s_{4}$, $s_{5}$ and $s_{6}$ are real roots. Further, by Routh-Hurwitz criterion, the necessary and sufficient conditions for $s_{i}(i=4,5,6)$ to lie in the left half plane are
\[a_{1}>0,\;\;a_{3}>0,\;\;a_{1}a_{2}>a_{3}.\]
That is to say, under the above conditions, the roots of $L_{1}(s)=0$ satisfy
\[|\arg (s_{i})|>\frac{\pi}{2}>\frac{\alpha\pi}{2}\;\;(i=4,5,6).\]
Therefore, $M_{*}$ is locally asymptotically stable and (i) holds

(ii)  By the results in \cite{14}, if $\mathnormal{D}_{1}(L_{1}(s))<0$, then $L_{1}(s)=0$ has a real root and a pair of conjugate complex roots, denoted by $s_{4}$, $m+ni$ and $m-ni$, respectively. Thus, one has
\[
\begin{split}
L_{1}(s)=&s^{3}+a_{1}s^{2}+a_{2}s+a_{3}\\=&(s-s_{4})(s-m-ni)(s-m+ni).
\end{split}
\]
By calculation,
 \[a_{1}=-s_{4}-2m,\;\;a_{2}=2s_{4}m+m^{2}+n^{2},\;\;a_{3}=-s_{4}(m^{2}+n^{2}).\]
 The conditions $a_{1}>0$, $a_{2}>0$ and $a_{3}>0$ imply that
\[-s_{1}>2m,\;\;m^{2}+n^{2}>-2s_{4}m,\;\;s_{4}<0.\]
Further, one has
\[m^{2}+m^{2}\tan^{2}\theta>-2s_{4}m>4m^{2}.\]
That is to say $\tan^{2}\theta>3$, which implies
\[\theta=|\arg(s)|>\frac{\pi}{3}.\]
Therefore, in order to ensure the establishment of $|\arg(s)|>\frac{\alpha\pi}{2}$, we must have $\alpha<\frac{2}{3}$. Thus (ii) holds.

The proof of conclusions (iii) and (iv) is similar to that of conclusion (ii), hence we omit it.
\end{proof}
\section{Numerical simulations}
\label{sec:8}
\subsection{Data Sources}
\label{sec:9}
The data used in this paper are from the Johns Hopkins University Center for Systems Science and Engineering (https://github.com/CSSEGISandData/COVID-19). the Johns Hopkins University publishes data of accumulated and newly confirmed cases, recovered cases and death cases of COVID-19 from January 22, 2020.
\subsection{Analysis of the SEIQRP Model}
\label{sec:10}
In the following discussion, the standard incidence rate \cite{18} is used to describe the transmission of COVID-19, and is given by
\[\beta_{1}f_{1}(S)g_{1}(I)=\frac{\beta_{1}SI}{N},\;\;\beta_{2}f_{2}(S)g_{2}(E)=\frac{\beta_{2}SE}{N},\]
where $N$ represent the total population of the region at the initial time.
Hence,
 \begin{eqnarray}
 \left\{\begin{split}
D^{\alpha}S(t)=&\Lambda-\frac{\beta_{1}S(t)I(t)}{N}-\frac{\beta_{2}S(t)E(t)}{N}-\mu S(t)-\rho S(t),\\
D^{\alpha}E(t)=&\frac{\beta_{1}S(t)I(t)}{N}+\frac{\beta_{2}S(t)E(t)}{N}-\epsilon E(t)-\mu E(t),\\
D^{\alpha}I(t)=&\epsilon E(t)-(\delta+\mu)I(t),\\
D^{\alpha}Q(t)=&\delta I(t)-(\lambda+\kappa+\mu)Q(t),\\
D^{\alpha}R(t)=&\lambda Q(t)-\mu R(t),\\
D^{\alpha}P(t)=&\rho S(t)-\mu P(t).
  \end{split}
 \right.
\end{eqnarray}
The effectiveness of the model (15) in describing the spread of COVID-19 is illustrated by selecting the confirmed, cured and dead cases in the United States. According to the real data provided by the Johns Hopkins University, the outbreak in the United States has not been brought under full control. The data of confirmed cases, cured cases and death cases are selected from January 22, 2020, to April 20, 2020. Assuming that the confirmed individuals will be isolated, then
\begin{equation}
Isolated=Confirmed-Recovered-Death.
\end{equation}
This hypothesis is in line with the actual situation. Hence, we can obtain the real data of isolated cases. Through the fractional predictor corrector method and the least squares fitting\cite{24}, we can identify the parameters of the model (15) through the real data, which can be found in Table 1.

\renewcommand\tablename{Table}
   \begin {table}[htbp]
    	\centering \caption{Summary table of the parameter identification for model (15) after using least squares fitting to real data from January 22, 2020, to April 20, 2020}
\begin {tabular}{c c}
	  \toprule
	   Notation & Parameter Identification \\ \midrule
      $\alpha$ &  0.7182 \\
	   $\Lambda$ &  3481608 \\
	   $\beta_{1}$ &  0.3813 \\
       $\beta_{2}$ &  0.7065 \\
        $\mu$ &  $6.7063\times10^{-8}$ \\
       $\rho$  &  0.1927 \\
       $\epsilon$  &  0.2657 \\
        $\delta$  &  0.3352 \\
       $\lambda$  &  0.0149 \\
       $\kappa$  &  $1.1975\times10^{-6}$ \\
	  \bottomrule							
\end{tabular}
\end{table}
Based on the parameters in Table 1, we can make a simple prediction of isolated cases and recovered cases in the United States, which can be found in Figure 1.  We need to emphasize that the peak here represents the number of isolated cases rather than the cumulative number of confirmed cases.
\begin{figure}
    \begin{center}
        \includegraphics[width=1\linewidth]{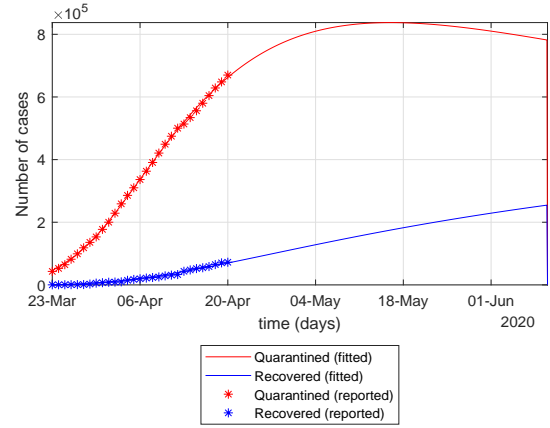}
       \caption{Number of isolated cases predicted and recovered cases predicted by the model (15) for the United States.}
    \end{center}
    \end{figure}
Under the parameters in Table 1, it can be calculated by using (9) that $\mathnormal{R}_{0}=0.2041<1$. By the conclusion in Theorem 3, the disease-free equilibrium point is local asymptotic stable. The SEIQRP model has a basic guiding significance for predicting and fitting spreading dynamics of COVID-19. However, the prediction of this model for COVID-19 is relatively rough, we still need to improve model (15) according to actual mitigation policies and research objectives. According to the analysis in Section 2, we choose the SEIQRPD model to predict the trend of the epidemic in the United States under reasonable assumptions.
\subsection{The SEIQRPD Model for the Prediction of COVID-19}
\label{sec:11}
Similarly, the standard incidence rate is used to describe the transmission of COVID-19, the fractional SEIQRPD model can be expressed as
 \begin{eqnarray}
 \left\{\begin{split}
D^{\alpha}S(t)=&-\frac{\beta_{1}(t)S(t)I(t)}{N}-\frac{\beta_{2}S(t)E(t)}{N}-\rho S(t),\\
D^{\alpha}E(t)=&\frac{\beta_{1}(t)S(t)I(t)}{N}+\frac{\beta_{2}S(t)E(t)}{N}-\epsilon E(t),\\
D^{\alpha}I(t)=&\epsilon E(t)-\delta I(t),\\
D^{\alpha}Q(t)=&\delta I(t)-(\lambda(t)+\kappa(t))Q(t),\\
D^{\alpha}R(t)=&\lambda(t) Q(t),\\
D^{\alpha}P(t)=&\rho S(t),\\
D^{\alpha}D(t)=&\kappa(t) Q(t).
  \end{split}
 \right.
\end{eqnarray}

When $\alpha=1$, the fractional-order SEIQRPD model is similar to the integer-order model used in \cite{24}. According to the data provided by Johns Hopkins University, by April 20, 2020, the outbreak in China has been basically controlled. In many provinces of China, the number of new cases per day is increasing in single digits. This means that the data in China contains more information about the spreading dynamics of COVID-19. Therefore, the data in Hubei, Guangdong, Hunan and Zhejiang are selected to research the fitting effect of the model (17). According to the real data of these four regions in China, the parameters of the model (17) are identified, and the results are shown in Table 2. The model (17) successfully capture the trend of the outbreak, which can be seen in Figure 3.
\begin{figure}[H]
    \begin{center}
        \includegraphics[width=1\linewidth]{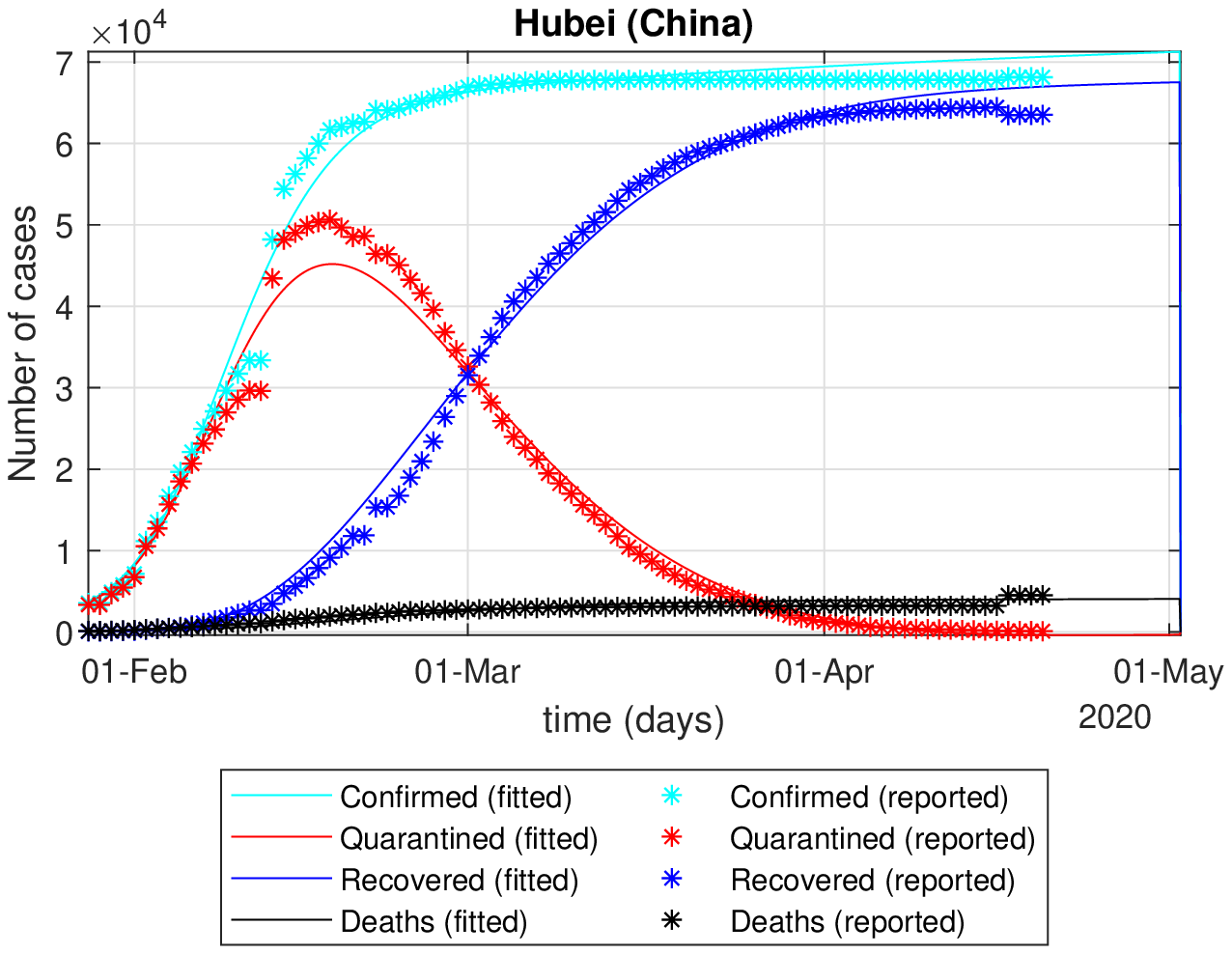}

    \end{center}
\vfill
    \begin{center}
        \includegraphics[width=1\linewidth]{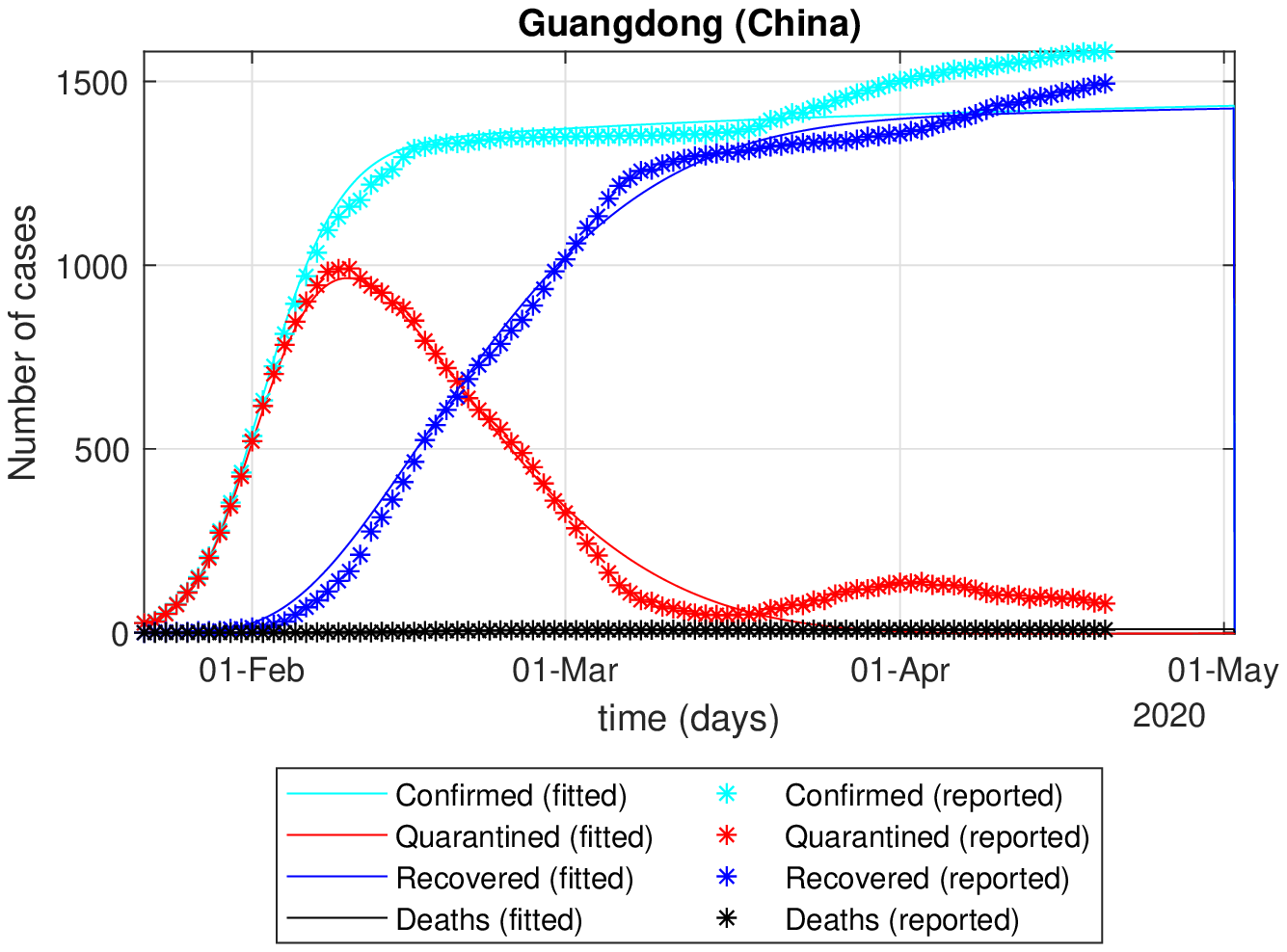}

    \end{center}
\end{figure}
\begin{figure}[H]
    \begin{center}
        \includegraphics[width=1\linewidth]{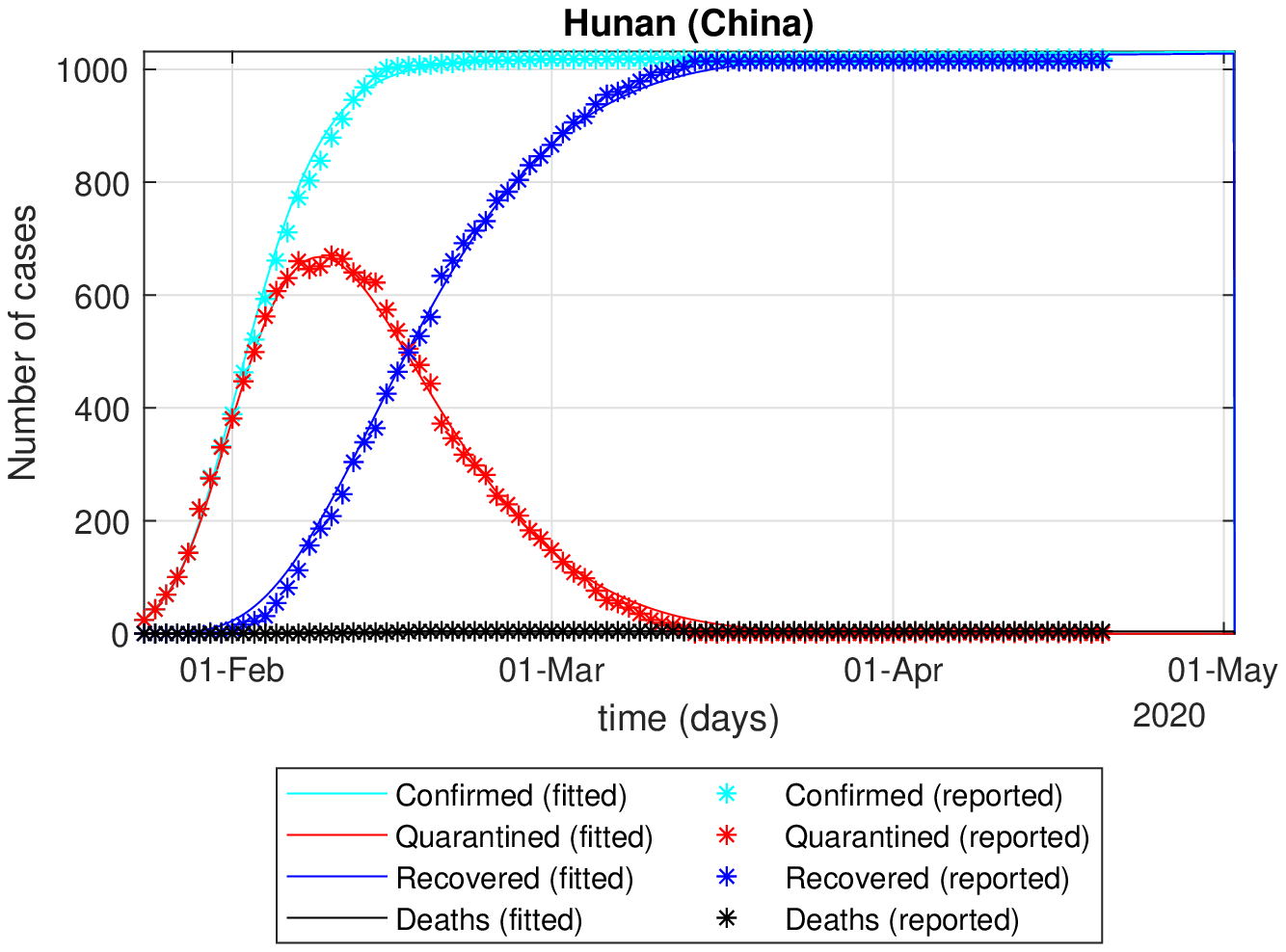}
    \end{center}
\vfill
    \begin{center}
        \includegraphics[width=1\linewidth]{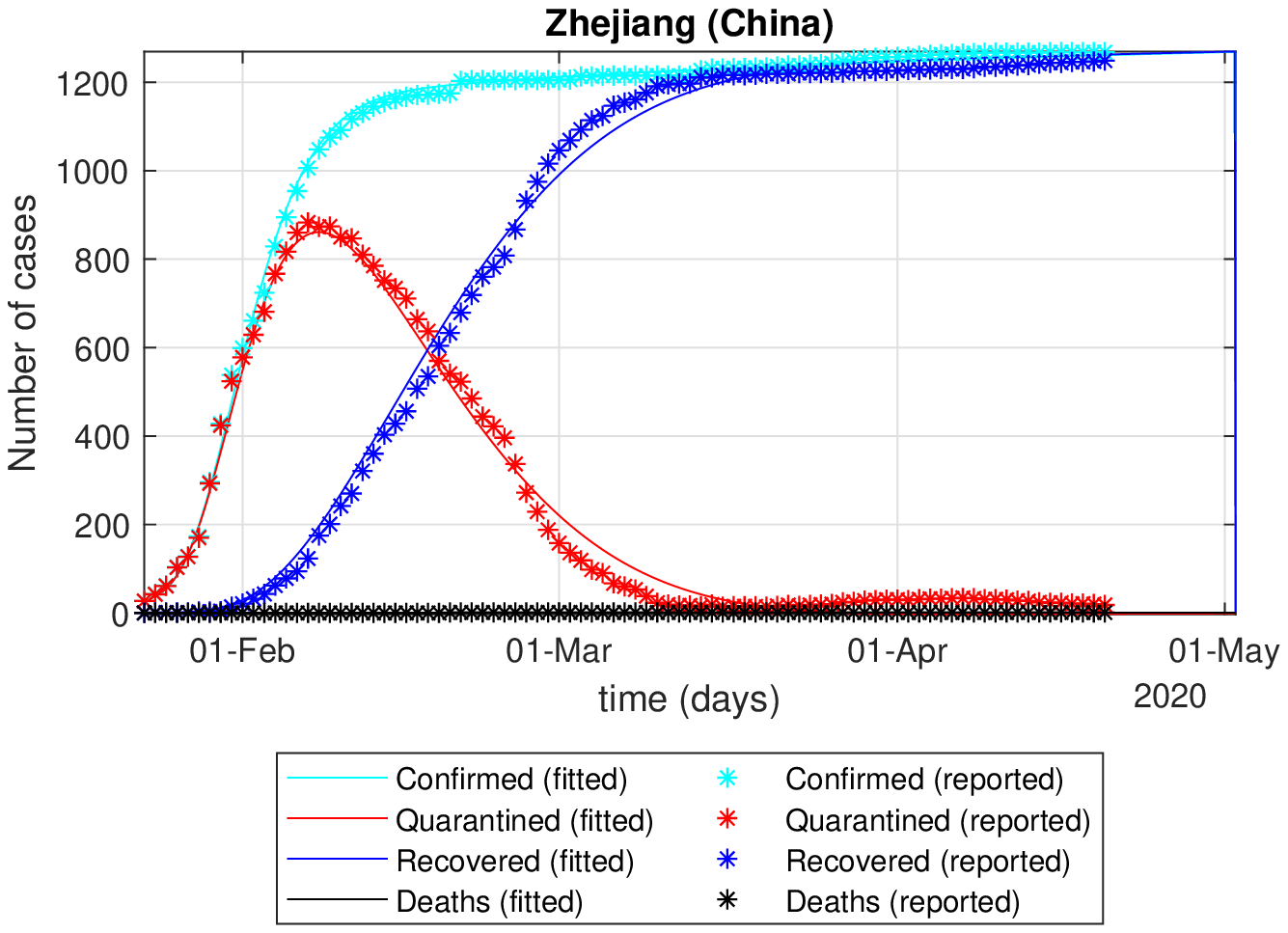}
       \caption{The fitting effect of the improved model (17) on the outbreak in Hubei, Guangdong, Hunan and Zhejiang.}
    \end{center}
    \end{figure}
 \renewcommand\tablename{Table}
   \begin {table}[H]
    	\centering \caption{Summary of parameter identification of model (17) (data used from January 22, 2020, to April 20, 2020).}
\begin {tabular}{c c c c c}
	  \toprule
	   Notation & Hubei & Guangdong& Hunan& Zhejiang\\ \midrule
      $\alpha$ &  1.0729 & 1.0446& 1.0135& 1.0482\\
	   $\sigma_{1}$ &  0.3954 &0.7057 & 1.4004&1.9985\\
	   $\sigma_{2}$ &  0.1345 &0.0783 & 0.4752&0.4592\\
       $\beta_{2}$ &  0.4132 &0.6594 & 0.6908&0.7608\\
       $\rho$  &  0.0913 &0.1033 & 0.0702& 0.1611\\
       $\epsilon$  &  0.2595 &0.4095 & 0.4264&0.3457\\
        $\delta$  &  0.2283 &0.3796 & 0.9589&0.3665\\
       $\lambda_{1}$  &  0.6509 &0.9695 & 0.9989&0.9142\\
       $\lambda_{2}$  &  0.002 &0.0019 & 0.0032&0.0025\\
       $\kappa_{1}$  &  0.002 & $2.8571\times10^{-4}$ & 0.0012&0.0098\\
       $\kappa_{2}$  & $6.3088\times10^{-6}$ & $2.7547\times10^{-14}$ & 0.1017&1.0968\\
	  \bottomrule							
\end{tabular}
\end{table}

At present, we can obtain the epidemic data of the United States from January 22, 2020 to April 20, 2020. We need to preprocess the data to remove the data smaller than 0.5\% of the current highest number of confirmed cases. In order to test the prediction ability of the SEIQRPD model (17) for the development process of the epidemic in the United States, we select the data before April 5 to identify the parameters of the SEIQRPD model (17).  Furthermore, in order to illustrate the ability of the SEIQRPD model (17) in predicting the outbreak, we compared the real data and fitted data after April 5, which can be seen in Table 3, Table 4 and Figure 4. According to the results in Table 3 and Table 4, it can be found that the real values of current isolated cases and cumulative confirmed cases fall within the range of 95\% - 105\% of the predicted values. This shows that the SEIQRPD model (17) can effectively predict the data in the next two weeks.

 \renewcommand\tablename{Table}
   \begin {table}[htbp]
    	\centering \caption{Summary of real and fitted data for the cumulative confirmed cases in the United States from April 6, 2020 to April 20, 2020.}
\begin {tabular}{c c c c c}
	  \toprule
\multirow{2}{*}{Data}&\multicolumn{2}{c}{Cumulative confirmed cases}&\multirow{2}{*}{Relative error (\%)}\\
\cmidrule(lr){2-3}
	   &Reported & Predicted \\ \midrule
        6, April & 366667 & 370220 & 0.97\\
	    7, April & 397505 & 403441 & 1.49\\
	    8, April & 429052 & 436940 & 1.84\\
        9, April & 462780 & 470492 & 1.67\\
        10, April & 496535 & 503882 & 1.48\\
        11, April & 526396 & 536905 & 2\\
        12, April & 555313 & 569374 & 2.53\\
        13, April & 580619 & 601117 & 3.53\\
        14, April & 607670 & 631983 & 4\\
        15, April & 636350 & 661839 & 4.01\\
        16, April & 667592 & 690572 & 3.44\\
        17, April & 699706 & 718090 & 2.63\\
        18, April & 732197 & 744318 & 1.66\\
        19, April & 759086 & 769200 & 1.33\\
        20, April & 784326 & 792699 & 1.07\\
	  \bottomrule							
\end{tabular}
\end{table}
 \renewcommand\tablename{Table}
   \begin {table}[htbp]
    	\centering \caption{Summary of real and fitted data for the isolated cases in the United States from April 6, 2020 to April 20, 2020.}
\begin {tabular}{c c c c c}
	  \toprule
\multirow{2}{*}{Data}&\multicolumn{2}{c}{Isolated cases}&\multirow{2}{*}{Relative error (\%)}\\
\cmidrule(lr){2-3}
	   &Reported & Predicted \\ \midrule
        6, April &  336303 &  339897 & 1.07\\
	    7, April &  362948 &  368221& 1.45\\
	    8, April &   390798&  396348& 1.42\\
        9, April &   420826&  424049& 0.77\\
        10, April &  449159&  451108& 0.43\\
        11, April &  474664&  477327& 0.56\\
        12, April &  500306&  502526& 0.44\\
        13, April &  513609&  526545& 2.52\\
        14, April &  534076&  549250& 2.84\\
        15, April &  555929&  570525& 2.63\\
        16, April &  579973 & 590279& 1.78\\
        17, April &  604388&  608442& 0.67\\
        18, April &  628693&  624968& 0.59\\
        19, April &  648088&  639826& 1.27\\
        20, April &  669903&  653007& 2.52\\
	  \bottomrule							
\end{tabular}
\end{table}
        \begin{figure}
    \begin{center}
        \includegraphics[width=1\linewidth]{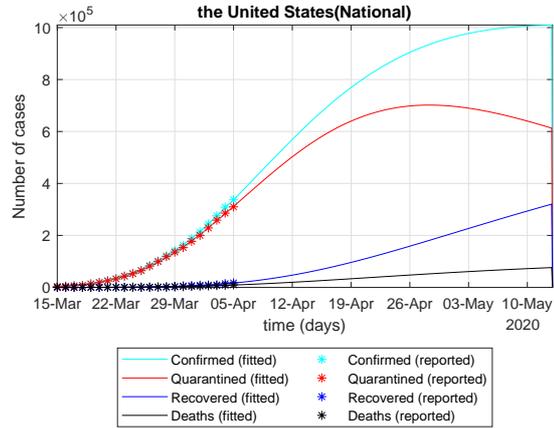}
       \caption{Based on the data of the United States from January 22 to April 20, 2020 to verify the accuracy of the forecast for the next 15 days.}
    \end{center}
    \end{figure}
The data before April 20, 2020 are selected to identify the parameters of the improved model (17), and the results are shown in Table 5. The cumulative number of confirmed cases and the number of isolated cases after two weeks are predicted, which can be seen in Table 6.
        \begin{figure}
    \begin{center}
        \includegraphics[width=1\linewidth]{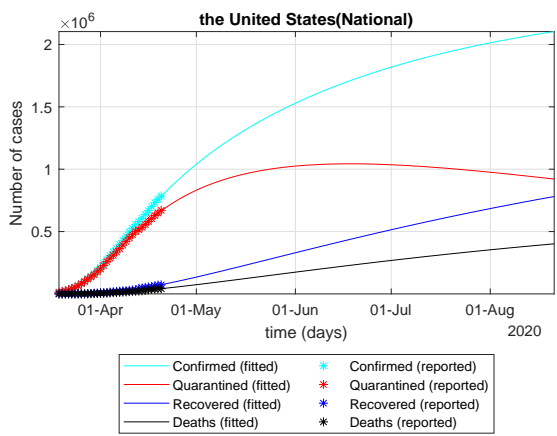}
       \caption{Based on the data of the United States from January 22 to April 20, 2020, to verify the accuracy of the forecast for the next 15 days.}
    \end{center}
    \end{figure}
    \renewcommand\tablename{Table}
   \begin {table}[H]
    	\centering \caption{Summary table of the parameter identification for model (17) (data used from January 22, 2020, to April 20, 2020)}
\begin {tabular}{c c}
	  \toprule
	   Notation & Parameter Identification \\ \midrule
       $\alpha$ &  0.8477\\
	   $\sigma_{1}$ & 0.8601  \\
	   $\sigma_{2}$ &  0.1553\\
       $\beta_{2}$ &  0.3659\\
       $\rho$  & $4.3455\times10^{-6}$\\
       $\epsilon$  &  0.3803\\
        $\delta$  &  0.4024\\
       $\lambda_{1}$  &  0.0121\\
       $\lambda_{2}$  &  0.1321\\
       $\kappa_{1}$  &  0.0064\\
       $\kappa_{2}$  & $3.7372\times10^{-4}$\\
	  \bottomrule							
\end{tabular}
\end{table}
 \renewcommand\tablename{Table}
   \begin {table}[H]
    	\centering \caption{Summary of predicted data for the United States from April 21, 2020 to May 5, 2020.}
\begin {tabular}{c c c}
	  \toprule
\multirow{2}{*}{Data}&\multicolumn{1}{c}{Cumulative confirmed cases}&\multicolumn{1}{c}{Isolated cases}\\
	   &Predicted &  Predicted \\ \midrule
        21, April & 807947 &  684148 \\
	    22, April & 833465 &  701725\\
	    23, April & 858444 &  718615\\
        24, April & 882890 &  734836\\
        25, April & 906816 &  750409\\
        26, April & 930232 &  765352\\
        27, April & 953150 &  779686\\
        28, April & 975584 &  793432\\
        29, April & 997546 &  806610\\
        30, April & 1019050 &  819239\\
        1, May & 1040108 &  831340\\
        2, May & 1060734 &  842931\\
        3, May & 1080941 &  854031\\
        4, May & 1100740 &  864658\\
        5, May & 1120145 &  874829\\
	  \bottomrule							
\end{tabular}
\end{table}
The isolated cases in the United States will peak on June 18, 2020, with the peak of $1.0431\times10^{6}$, which can be seen in Figure 5.
\section{Conclusion}
\label{sec:12}
We first propose the fractional SEIQRP model with generalized incidence rates. Some qualitative properties of the SEIQRP model are discussed. In order to predict COVID-19 effectively, we propose an improved SEIQRPD model according to the actual mitigation situation. According to the data of the United States before April 5, 2020, the trend of the outbreak in the United States from April 6 to April 20 is successfully predicted as compared to the real records. Then, using the data before April 20, 2020, we forecast the trend of the outbreak in the United States in the next two weeks, and estimate the peak of isolated cases and the date of the peak.

The improved SEIQRP model proposed in this paper successfully captures the trend of COVID-19. The long-term prediction needs to adjust the model appropriately according to the change of policy and medical level. We will discuss in the future work.

\section*{Acknowledgment}

The plots in this paper were plotted using the plot code adapted from \cite{24}.

\section*{Conflict of Interest}

We declare that we have no conflict of interest.






\begin{thebibliography}{00}


\bibitem{1}  Lin, Q.Y., Zhao, S. Gao, D.Z., Luo, Y.J., Yang, S., Musa. S.S., Wang, M.H., Cai, Y.L., Wang, W.M., Yang, L., He, D.h.: A conceptual model for the coronavirus disease 2019 (COVID-19) outbreak in Wuhan, China with individual reaction and governmental action. International Journal of Infectious Diseases. \textbf{93}, 211-216 (2020).


\bibitem{2}
Peng, L., Yang, W., Zhang, D., Zhuge, C., Hong, L.: Epidemic analysis of COVID-19 in China by dynamical modeling. Cold Spring Harbor Laboratory. (2020).
arXiv: 2002.06563.


\bibitem{3}
Amjad, S. S., Iqbal, N. S., Kottakkaran, S.N.: A Mathematical model of COVID-19 using fractional derivative: outbreak in India with dynamics of transmission and control. (2020). doi: 10.20944/preprints202004.0140.v1.


\bibitem{4} Chen, Y., Cheng, J., Jiang, X., Xu, X.: The reconstruction and prediction algorithm of the fractional TDD for the local outbreak of COVID-19. (2020). arXiv:2002.10302.


\bibitem{5}
Cheng, Z.J., Shan, J.: 2019¨Cnovel Coronavirus: Where We are and What We Know. Infection. (2020). DOI 10.1007/s15010-020-01401-y.

\bibitem{6} Yang, Z.F., Zeng, Z.Q., Wang, K., Wong, S., Liang, W.H., Zanin, M., Liu, P., Cao, X.D., Gao, Z.Q., Mai, Z.T., Liang, J.Y., Liu, X.Q., Li, S.Y., Li, Y.M., Ye, F., Guan, W.J., Yang, Y.F., Li, F., Luo, S.M., Xie, Y.Q., Liu, B., Wang, Z.L., Zhang, S.B., Wang, Y.N., Zhong, N.S., He, J.X.: Modified SEIR and AI prediction of the epidemics trend of COVID-19 in China under public health interventions. Journal of Thoracic Disease \textbf{12}(2), 165-174 (2020).

\bibitem{7}
Zhao, S., Lin, Q, Ran J., Musa, S.S., Yang, G.P., Wang, W.M., Lou, Y.J., Gao, D.Z., Yang, L., He, D.H., Wang, M.H.: Preliminary estimation of the basic reproduction number of novel coronavirus (2019-nCoV) in China, from 2019 to 2020: A data-driven analysis in the early phase of the outbreak. International Journal of Infectious Diseases. 214-217 (2020).

\bibitem{8} Podlubny, I.:  Fractional differential equations, Academic Press, New York, (1999).

\bibitem{9}
Diethelm, K.:  The analysis of fractional differential equations: an application-oriented exposition using differential operators of Caputo type, Springer Science, Business Media, (2010).

\bibitem{10}
Kuniya, T.: Hopf bifurcation in an age-structured SIR epidemic model. Applied Mathematics Letters. \textbf{92}, 22-28 (2019).

\bibitem{11}
Zhang, X.B., Huo, H., Xiang, H.F., Xiang, H., Meng, X.Y.: An SIRS epidemic model with pulse vaccination and non-monotonic incidence rate. Nonlinear Analysis: Hybrid Systems. 13-21 (2013).


\bibitem{12} Cai, Y., Kang, Y., Wang, W.: A stochastic SIRS epidemic model with nonlinear incidence rate. Applied Mathematics and Computation. 221-240 (2017).
\bibitem{13} Ahmed, E., Elgazzar, A.S.: On fractional order differential equations model for nonlocal epidemics. Physica A. \textbf{379}(2), 607-614 (2007).
\bibitem{14} Rocca, A., West, B.J.: Fractional calculus and the evolution of fractal phenomena. Physica A. \textbf{265}(3-4), 535-546 (1999).
\bibitem{15} Jalilian, Y., Jalilian, R.: Existence of solution for delay fractional differential equations. Mediterranean Journal of Mathematics. \textbf{10}(4), 1731-1747 (2013).
\bibitem{16} Yang, Y., Xu, L.: Stability of a fractional order SEIR model with general incidence. Applied Mathematics Letters, \textbf{105}, (2020).
\bibitem{17} Wang, H., Yu, Y.G., Wen, G.G., Zhang, S., Yu, J.Z.: Global stability analysis of fractional-order Hopfield neural networks with time delay. Neurocomputing. \textbf{154}, 15-23, (2015).
\bibitem{18} Hu, Z., Liu, S., Wang, H.: Backward bifurcation of an epidemic model with standard incidence rate and treatment rate. Nonlinear Analysis-real World Applications. \textbf{9}(5), 2302-2312 (2008).
\bibitem{19}
    Sun, H.G., Zhang, Y., Baleanu, D., Chen, W., Chen, Y.Q.: A new collection of real world applications of fractional calculus in science and engineering. Communications in Nonlinear Science and Numerical Simulation. \textbf{59}(5), 213-231 (2018).
\bibitem{20}
    Cao, K.C., Chen, Y.Q.: Fractional order crowd dynamics: Cyber-Human systems modeling and control. (Invited book project. Volume 4 of the De Gryuter Monograph Series `Fractional Calculus in Applied Sciences and Engineering`). ISBN 978-3-11-047398-8.

\bibitem{21}
    West, B.J.: Fractional Calculus View of Complexity: Tomorrow's Science, CRC Press, (2015).
\bibitem{22}
    Li, Y., Chen, Y.Q, Podlubny, I.: Mittag-Leffler stability of fractional order nonlinear dynamic systems. Automatica. \textbf{45}(8), 1965-1969 (2009).
\bibitem{23} Lu, Z.Z., Yu, Y.G., Chen, Y.Q., Ren, G.J., Xu, C.H., Yin, Z.: A fractional-order SEIHDR model for COVID-19 with inter-city networked coupling effects. Nonlinear Dynamics (Special Issue on `Nonlinear dynamics of COVID-19 pandemic: modeling, control, and future perspectives`). (April 2020).
\bibitem{24} Cheynet, E.: Generalized SEIR epidemic model (fitting and computation). (https://www.github.com/ECheynet/SEIR), GitHub.
\end{thebibliography}

\end{document}